\documentclass[aip,amsmath,amssymb]{revtex4-1}

\usepackage[utf8]{inputenc}
\usepackage[T1]{fontenc}
\usepackage{mathptmx}
\usepackage{amsfonts}
\usepackage{amsthm}
\usepackage{graphicx}
\usepackage{bm}
\usepackage{xcolor}
\usepackage{longtable}
\usepackage{float}
\usepackage{tikz}
\usepackage[shortlabels]{enumitem}
\usepackage{hyperref}
\usepackage{etoolbox}

\graphicspath{{./}}

\newtheorem{theorem}{Theorem}

\newtheorem{corollary}{Corollary}
\newtheorem{definition}{Definition}

\newtheorem{remark}{Remark}

\def\LL{\mathcal{L}}

\makeatletter
\def\@email#1#2{%
 \endgroup
 \patchcmd{\titleblock@produce}
  {\frontmatter@RRAPformat}
  {\frontmatter@RRAPformat{\produce@RRAP{*#1\href{mailto:#2}{#2}}}\frontmatter@RRAPformat}
  {}{}
}%
\makeatother

\begin{document}

\title{Coupling of diffusion and reaction in a thin cylindrical tube: Methodological drawbacks of the Fick--Jacobs approach}

\author{Sergey~D.~Traytak}
\email{sergtray@mail.ru}
\affiliation{Semenov Federal Research Center for Chemical Physics, Russian Academy of Sciences, 4 Kosygina St., Moscow 119991, Russian Federation}
% ORCID: 0000-0002-3551-8796

\author{Timofey~V.~Fedoseev}
\email{fedoseev.tv@phystech.edu}
\affiliation{Semenov Federal Research Center for Chemical Physics, Russian Academy of Sciences, 4 Kosygina St., Moscow 119991, Russian Federation}
\affiliation{Moscow Institute of Physics and Technology, 9 Institutskyi Lane, Dolgoprudnyi 141701, Moscow Region, Russian Federation}
% ORCID: 0009-0004-1506-4186

\date{\today}

\begin{abstract}
We investigate a problem, that describes coupling between diffusion and reaction inside a thin circular cylindrical tube.  The asymptotic solution of the posed  problem is derived by means of the boundary functions method.  We perform comparison of this asymptotic solution against corresponding exact solution, which revealed serious methodological drawbacks of known Fick-Jacobs reduction approach. The results obtained may be used to study a wide range of reaction-diffusion problems, when the Fick-Jacobs method cannot be applied.
\end{abstract}

\maketitle

%%%%%%%%%%%%%%%%%%%%%%%%%%%%%%%%%%%%%%%%%%%%%%%%%%%%%%%%
\section{Introduction}\label{sec:Intro}
%%%%%%%%%%%%%%%%%%%%%%%%%%%%%%%%%%%%%%%%%%%%%%%%%%%%%%%

Over the last forty years since Zwanzig seminal papers~\cite{Zwanzig83,Zwanzig92},
the solution of the diffusion equation inside a tube of varying cross section has seen a rapid growth of interest and activity. A large number of analytical and numerical results for this and appropriate problems was obtained up to now. However we do not intend to present here any survey of relevant publications in this field. Interested readers can find rather extensive historical review on the idea of the so-called {\it Fick-Jacobs 1D reduction} and current state of the issue in an exhaustive book~\cite{Dagdug24} (see also recent papers~\cite{Chacón-Acosta23,Krasowska25,Burada26,Nayak26,Dagdug26}). Besides, it should be highlighted the clear and comprehensive development of the relevant diffusion problem with the aid of kinetic theory performed in Ref.~\cite{Brey24}.

Here we briefly dwell on necessity and reasons of present investigation. First of all highlight that in our previous article~\cite{Traytak14} by means of matched asymptotic expansions approach we obtained the uniformly valid leading-term approximation to solution of the 3D diffusion problem in a slender tube of varying cross section.
Particularly it was shown out there that the use of the singular perturbations theory had directly come up with the   Fick-Jacobs 1D reduction: it had used a method from related fields of applied mathematics and mechanics that no one had previously thought of applying to this type of physical problem. The task has been solved rigorously and besides the mathematical theory presented in details we provided there a number of physical examples. So it seemed that at least the questions about the methodological approach were closed. Nevertheless, the above rigorous asymptotic method appeared to be rather sophisticated with technically complicated implementation procedure and apparently, therefore, it has not received due attention.

Thus, one of the main incentive for writing this paper was an attempt to consider here the relevant so-called {\it toy model}. By this we mean a greatly simplified diffusion model, which helps us to elucidate the singular perturbation methods understanding the principal techniques of approaching more complex problem. The choice of the specific toy model was primarily dictated by its simplicity for clear asymptotic analysis and then existing of the exact analytical solution.
That is why among different singular perturbation methods, their versions and implementations, we will utilize the so-called {\it method of boundary functions}, which is the generally recognized as the simplest one~\cite{Butuzov87}.
To make the task analytically tractable and clear asymptotic solution procedure we will investigate a simple-minded physical model considering very small Brownian particles, diffusing inside a finite circular cylindrical tube with source of particles prescribed on the wall and one end towards perfectly absorbing end.

Summarising we conclude that our goal is twofold: (1) show how simple and naturally it may be solved by known asymptotic method and (2) point out to serious difficulties of commonly used the Fick–Jacobs method even in this simple toy problem.

The paper is arranged as follows.
In next Sec.~\ref{sec:Statement} we rigorously formulate the boundary value problem for the steady-state diffusion equation to calculate the local concentration and reaction rate within the scope of the arbitrary circular cylinder model.  Section~\ref{sec:Exact} contains exact analytical solution of the posed problem obtained with the help of standard method of separation of variables. Here we also found the corresponding asymptotic expansion of the exact solution for the case of thin cylinder. Formulation of the problem under consideration as a singular perturbed problem is given in Sec.~\ref{sec:singular}. We study the spectrum of the unperturbed operator by means of the Vishik-Lyusternik theory and proved that above singular perturbed problem is not on the spectrum. This result allows us to apply very powerful boundary functions technique to derive explicit asymptotic solution in Sec.~\ref{sec:Asymptotic}. Sec.~\ref{sec:1Dreduction} contains some discussion concerning essential difficulties in applying the method of the Fick-Jacobs 1D reduction to describe diffusion-controlled reaction inside a thin circular cylinder. Finally, conclusions together with possible future extension of the present work are summarized in Sec.~\ref{sec:Conclusion}.

%%%%%%%%%%%%%%%%%%%%%%%%%%%%%%%%%%%%%%%%%%%%%%%%%%%%%%%%
\section{Statement of the problem} \label{sec:Statement}
%%%%%%%%%%%%%%%%%%%%%%%%%%%%%%%%%%%%%%%%%%%%%%%%%%%%%%%

%%%%%%%%%%%%%%%%%%%%%%%%%%%%%%%%%%%%%%%%%%%%%%%%%%%%%%%%%%%%%%%%%%%%%%%%%%%
\subsection{Physical background}  \label{subsec:Backappendix}
%%%%%%%%%%%%%%%%%%%%%%%%%%%%%%%%%%%%%%%%%%%%%%%%%%%%%%%%%%%%%%%%%%%%%%%%%%%

Let us study the free steady-state Brownian motion of very small entities (point-like particles) ($B$-particles) into a 3D bounded region with an absorbing site located on its boundary.
For the sake of simplicity we have limited our treatment to reaction-diffusion processes in homogeneous and isotropic media, 
paying attention to the {\it restricted geometry effects}.

In case of diffusion-controlled reactions we develop both exact and asymptotic formalisms, based on Smoluchowski-type approach within this formalism we obtain explicit expressions for the local concentration and reaction rate.

It has been realized that very often irreversible diffusion-controlled reactions inside a bounded domain $\Omega \subset \mathbb{R}^3$  occur upon its chemically heterogeneous boundary $\partial \Omega$ described by the simplest reaction scheme:~\cite{Rice85}
\begin{align}
	S+B\;{\rightarrow}\; S + Products\,.\label{DCR1}
\end{align}
For instance active site $S$ may be treated as some catalytic substrate consists immobile catalytic site located on the an active part of the domain $\Omega$ boundary $\partial\Omega$~\cite{Aris95}.

In this paper we will consider the diffusion problem for the case of reactions with the full diffusion control posed in a cylindrical tube of circular cross-section. As regards adopted notations and definitions, here, we basically follow known paper~\cite{Dagdug25}, which deals with the trapping of a particle diffusing in a cylindrical cavity by circular absorbing spots of arbitrary radii located in the centers of the cavity bases.

%%%%%%%%%%%%%%%%%%%%%%%%%%%%%%%%%%%%%%%%%%%%%%%%%%%%%%%%%%%%%%%%%%%%%%%%%%%
\subsection{Geometrical description}  \label{subsec:Geometry}
%%%%%%%%%%%%%%%%%%%%%%%%%%%%%%%%%%%%%%%%%%%%%%%%%%%%%%%%%%%%%%%%%%%%%%%%%%%

Assume that point-like $B$-particles diffuse inside a 3D circular cylindrical tube $\Omega $ (configuration manifold~\cite{Traytak24}) of radius $R$ and length $L$ ($R<L$) see Fig.~\ref{Fig1}. It should be emphasized that cylindrical geometry is very attractive for theoretical analysis due to its axial symmetry and consequent simplicity of mathematical description.

Taking the $z$-coordinate axis the same, it is expedient to transform the
Cartesian coordinate system $\left\{O;x,y,z\right\}$  to
circular cylindrical coordinates $\left\{O;r,\phi,z\right\}$ connected by
\begin{equation}
	x = r \cos  \phi  \,,\qquad
	y = r \sin \phi  \,,\qquad
	z = z \,.
	\label{CylC1}
\end{equation}
Clearly, for all azimuthal angles \(0<\phi<2\pi\), one has
\begin{align}
	\Omega = \left\{0<r=\sqrt{x^2+y^2}<R\right\}\times\left\{0<z<L\right\} \,.
	\label{CylC2}
\end{align}
Denote by $\Sigma_{z}$:  $ 0 \le z \le L$ the cylindrical tube cross-section at any given point $z$ of the symmetry axis and its area by $ \vert \Sigma_{z} \vert = \pi R^2$.

Corresponding cylindrical tube boundary $\partial\Omega$ (see Fig.~\ref{Fig1}) comprises: (i) the lateral cylindrical boundary 
\begin{align}
	\partial\Omega_w:=\left\{ r = R \right\}\times \left\{ 0 < z < L \right\}
	\label{WallB}
\end{align} 
called the wall, and (ii) two parts of the boundary 
\begin{eqnarray}
	\Sigma_0:=\left\{ z = 0 \right\}\times \left\{ 0 < r < R \right\}\,, \label{LeftB} \\ 
	\Sigma_L:=\left\{ z = L \right\}\times \left\{ 0 < r < R \right\}  \,, \label{rightB} 
\end{eqnarray}
which are termed the left and right cylindrical ends. 

Moreover, for the problem at issue it is convenient to decompose the cylindrical tube boundary as follows:
\begin{align}
	\partial\Omega: = \partial\Omega_s \cup \Sigma_L \,, \quad \partial\Omega_s:= \Sigma_0 \cup \partial\Omega_w\,.
	\label{Bounds}
\end{align}
Here $\partial\Omega_s$ and $\Sigma_L$ denote the part of the boundary with the source of $B$-particles and absorbing part, respectively. Thus the domain $\Omega$ is bounded by the $ 3 $-connected boundary $\partial\Omega$ with connected components $ \left\{ \Sigma_0, \partial\Omega_w, \Sigma_L\right\} $.

%%====================================================================================================================
\begin{figure}[htbp]
\centering
\includegraphics[width=0.5\textwidth]{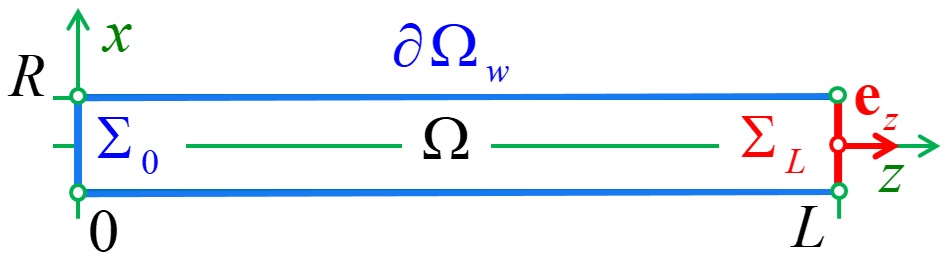}
\caption{The main cylindrical section along $Oy$ axis for a 3D cylindrical tube with radius $R$ and length $L$. We depicted the absorbing $\Sigma_L$ (\ref{rightB}) (red) and the source $\partial\Omega_s$ (\ref{Bounds}) (blue) parts of the boundary. The Cartesian coordinates are shown in green.}
\label{Fig1}
\end{figure} 
%%====================================================================================================================

%%%%%%%%%%%%%%%%%%%%%%%%%%%%%%%%%%%%%%%%%%%%%%%%%%%%%%%%%%%%%%%%%%%%%%%%%%%%%%%%
\subsection{Governing equations and boundary conditions}\label{subsec:Basic}
%%%%%%%%%%%%%%%%%%%%%%%%%%%%%%%%%%%%%%%%%%%%%%%%%%%%%%%%%%%%%%%%%%%%%%%%%%%%%%%%

To avoid unnecessary complications let us impose the axisymmetric boundary conditions on $ \partial\Omega$ making problem independent of the azimuthal angle $\phi$. Besides, in order to describe the above diffusion–reaction model mathematically it is expedient to introduce a {\it normalized dimensionless local concentration} given throughout the entire domain $\Omega$~\cite{Rice85,Traytak24}
\begin{align}
	u \left(r, z\right) := n_B \left(r, z\right)/c_B \in \left(0,\, 1 \right) \,, \label{SurPr1} 
\end{align}
where $ n_B \left(r, z\right)$ is the {\it local concentration} (number density) of $B$-particles and  its constant value $ c_B $ prescribed at the boundary $ \partial \Omega_s $. Further on, we will call all transformations of function (\ref{SurPr1}) just  local concentration if there is no confusion.

Thus within the above domain $\Omega$ we will treat the following steady-state Fick's system: continuity equation and the classical constitutive relation (the first Fick’s law of diffusion) defined as
\begin{eqnarray}
	\bm \nabla \cdot \mathbf{j} =0\,, \label{Lap0a} \\ 
	\mathbf{j} = -D{\bm \nabla} u  \,. \label{Lap1a} 
\end{eqnarray}
Hereafter
$ \mathbf{j} \left(r, z\right) $ is the local diffusive flux of  at a current point $ \left(r, z\right) \in \Omega $ and $D$ is the translational diffusion coefficient of $B$-particles, and $ \bm{\nabla} $ denotes the gradient operator in $\mathbb{R}^3$.

The Fick diffusive system Eqs. (\ref{Lap0a}), (\ref{Lap1a}) immediately leads to the equation governing local concentration 
\begin{align}
	-D{\bm \nabla}^2 u =0 \quad \mbox{in} \quad \Omega \,.	\label{Lap1}
\end{align}

Using the cylindrical coordinate system the steady-state diffusion equation (\ref{Lap1}) for
the local concentration $ u\left(r,z\right)$ it is convenient to put in the divergent form
\begin{align}
\frac{1}{r}
\frac{\partial}{\partial r}
\left(
r\frac{\partial u}{\partial r}
\right)
+
\frac{\partial^2 u}{\partial z^2}
=0 
\quad
\text{in} \quad \Omega \,.
\label{CcDE1}
\end{align}

This equation is subject to the corresponding Dirichlet boundary conditions on
$\partial\Omega $:
%Along with Eq.~(\ref{CcDE1}) we impose the following  boundary conditions 
\begin{eqnarray}
\left. u \right|_{\partial\Omega_s} = 1 \,, \quad\left. u \right|_{\Sigma_L} = 0 \,.
\label{Bc2}
\end{eqnarray}
\begin{remark}
We point out that the homogeneous Dirichlet boundary condition on $\Sigma_L$ is also known as Smoluchowski's boundary condition in applications~\cite{Rice85}.
\end{remark}
Thus, we deal with the internal Dirichlet boundary value problem for the Laplace equation (\ref{CcDE1}), (\ref{Bc2}) imposed in the cylindrical domain $\Omega$.

%%%%%%%%%%%%%%%%%%%%%%%%%%%%%%%%%%%%%%%%%%%%%%%%%%%%%%%%%%%%%%%%%%%%%%%%%%
\subsection{Dimensionless form of the problem} \label{sec:Dimensionless}
%%%%%%%%%%%%%%%%%%%%%%%%%%%%%%%%%%%%%%%%%%%%%%%%%%%%%%%%%%%%%%%%%%%%%%%%%%

Next, we cast the problem (\ref{CcDE1}), (\ref{Bc2}) in a dimensionless form using $R$ and $L$ as a reference length scales
\begin{eqnarray}
\rho = \frac{r}{R} \,,\quad \xi = \frac{z}{L} \,, \label{DimVar1}\\
\Omega \to \Omega_{\epsilon}:= \left\{0<\rho<1\right\}\times\left\{0<\xi<1\right\}\,.\label{DimOmega}
\end{eqnarray}
Then the dimensionless diffusion problem for the local concentration $u\left( \rho,\xi; \epsilon \right)$ becomes
\begin{eqnarray}
\frac{1}{\rho}
\frac{\partial}{\partial \rho}
\left( \rho\frac{\partial u}{\partial \rho}
\right)	+
	\epsilon^2\frac{\partial^2 u}{\partial \xi^2} = 0 \quad \text{in} \quad \Omega_{\epsilon} \,, 
	\label{BVP1} \\
	\left. u \right|_{\rho=1} = 1 \,,
	\label{BVP2}\\
	\left. u \right|_{\xi=0} = 1 \,, \quad \left. u \right|_{\xi=1} = 0 \,.  
	\label{BVP3}	
\end{eqnarray}
In Eq. (\ref{BVP1}) we introduced an important dimensionless {\it thickness parameter}
\begin{equation}
	0< \epsilon = {R}/{L}  \,,
	\label{RL1}
\end{equation}
which is determined by the geometry of the tube domain $\Omega_{\epsilon}$. At this stage of the study we do not impose any restrictions on the magnitude of the thickness parameter $\epsilon$ (\ref{RL1}).
\begin{remark}
Note that, to simplify notations, below we mostly omit the thickness parameter $\epsilon$ in function $u\left( \rho,\xi; \epsilon \right)$.
\end{remark} 

It is clear that on the $z$ axis one should also take into consideration conditions of regularity and axial symmetry of the solution~\cite{Traytak14}, respectively
\begin{eqnarray}
	\left. u \right|_{\rho=0} < \infty \,,
	\qquad
	\left. \frac{\partial u}{\partial \rho} \right|_{\rho=0} =  0  \,.
	\label{AxisRegExact1}
\end{eqnarray}

%%%%%%%%%%%%%%%%%%%%%%%%%%%%%%%%%%%%%%%%%%%%%%%%%%%%%%%%%%%%%%
\subsection{The trapping rate}\label{sec:rate}
%%%%%%%%%%%%%%%%%%%%%%%%%%%%%%%%%%%%%%%%%%%%%%%%%%%%%%%%%%%%%%%

According to the Smoluchowski theory~\cite{Rice85} in order to describe the kinetics of the irreversible diffusion-controlled reactions one should estimate the {\it trapping reaction rate} defined by
\begin{align}
	k \left( \epsilon\right):=\Phi_L \left( \epsilon\right)/c_B\,, \label{TotFl}
\end{align}
where $\Phi_L \left( \epsilon\right)$ stands for the total flux of $B$-particles on the right end of the cylinder $\Sigma_L$. 

Once the local concentration of $B$-particles $ u \left(r, z\right)$ is in hand, the  trapping reaction rate can be calculated straightforwardly by the general formula 
\begin{align}
	k \left( \epsilon\right) =\int\limits_{\Sigma_L}\left. \left( {\mathbf e}_z \cdot \mathbf{j} \right) \right| _{\Sigma_L}dS \,.  \label{Lap2}
\end{align}
Hereinafter $ {\mathbf e}_z \left(r, z\right) $ being the normal unit vector pointing outward of $ \Omega $ at its spatial point of the right cylindrical end $ \Sigma_L $ (see Fig.~\ref{Fig1}) and $d{S}$ is differential element of $\Sigma_L$ area. 

Then, recasting formula (\ref{Lap2}) with respect to $ u \left(r, z\right) $, one can easily obtain the desired trapping rate by the integral %$\partial\Omega_{\xi} $ (\ref{Bound2})
\begin{align}
	k \left( \epsilon\right) = - \pi RD\int_{0}^{1} \left. \frac{\partial u}{\partial \xi}\right| _{\xi=1} \rho d\rho \,.	\label{dm1n}
\end{align}
Moreover, it is convenient to introduce the dimensionless trapping rate into the absorbing end $k^{\ast}$ normalized $k$ (\ref{dm1n}) by $\pi RD$, i.e.
\begin{align}
k^{\ast} \left( \epsilon\right):=k \left( \epsilon\right)/\pi RD \,.	\label{dm1na}
\end{align}
The normalized trapping rate of $B$-particles on a given absorbing surface (\ref{dm1na}) is of primary importance for the theory of diffusion-controlled reactions.

%%%%%%%%%%%%%%%%%%%%%%%%%%%%%%%%%%%%%%%%%%%%%%%%%%%%%%%%
\section{Exact analytical solution}\label{sec:Exact}
%%%%%%%%%%%%%%%%%%%%%%%%%%%%%%%%%%%%%%%%%%%%%%%%%%%%%%%

In this Section we present exact analytical solution to the above posed diffusion problem (\ref{BVP1})-(\ref{BVP3}) in cases of arbitrary and small thickness parameter (\ref{RL1}).

%%%%%%%%%%%%%%%%%%%%%%%%%%%%%%%%%%%%%%%%%%%%%%%%%%%%%%%%%%%%%%
\subsection{The general case}\label{subsec:general}
%%%%%%%%%%%%%%%%%%%%%%%%%%%%%%%%%%%%%%%%%%%%%%%%%%%%%%%%%%%%%%%

Although the local concentration $u \left( \rho,\xi \right)$ (\ref{SurPr1}) may be used for further treatment,
below, we are going to utilize the so-called {\it complementary normalized local concentration}, which seems to be more
appropriate to describe the restricted geometry effects
\begin{align}
	v \left( \rho,\xi \right) = 1-u \left( \rho,\xi \right) \,.
	\label{ComplFunc1}
\end{align}
Then it is obvious that the function $v \left( \rho,\xi \right)$ (\ref{ComplFunc1}) obeys the boundary value problem:
\begin{eqnarray}
	\frac{1}{\rho}
	\frac{\partial}{\partial \rho}
	\left( \rho\frac{\partial v}{\partial \rho}
	\right)	+
	\epsilon^2\frac{\partial^2 v}{\partial \xi^2} = 0 \quad \text{in} \quad \Omega_{\epsilon} \,, 
	\label{BVPv1} \\
	\left. v \right|_{\rho=1} = 0 \,,
	\label{ComplBCwall1}\\
	\left. v \right|_{\xi=0} = 0 \,, \quad \left. v \right|_{\xi=1} = 1 \,.
	\label{ComplBCright1}
\end{eqnarray}
Following the standard method of separation of variables~\cite{Carslaw59}
we can arrive at the general form of solution to Eq. (\ref{BVP1})
\begin{align}
v \left( \rho,\xi \right) =
\sum_{m=1}^{\infty}
C_m
J_0 \left( \lambda_m\rho \right)
\sinh\left(\lambda_m\xi/\epsilon \right)\,, \label{GenSol1}
\end{align}
where $C_m$ are the unknown real coefficients to be determined from the boundary conditions (\ref{ComplBCwall1}), (\ref{ComplBCright1}),
and $J_\nu\left(\lambda_m\rho\right)$  is the Bessel function of the first kind of $\nu$-th order (see Appendix).
Hereafter the real numbers $\lambda _{k}$ are the roots of the transcendental equation (\ref{Hs4}).

Using boundary condition (\ref{ComplBCright1}) at $\xi =1$ and orthogonality condition (\ref{PolL3}) we readily obtain the explicit expression for the exact solution to the problem (\ref{BVPv1})-(\ref{ComplBCright1})
\begin{align}
	v \left( \rho,\xi \right)
	=2\sum_{m=1}^{\infty}
	\frac{J_0 \left( \lambda_m\rho \right)}
	{\lambda_m J_1 \left( \lambda_m \right)}
	\frac{\sinh\left( \lambda_m\xi/\epsilon \right)}
	{\sinh\left( \lambda_m/\epsilon \right)} \,.
	\label{wFinalExact1}
\end{align}
With the help of Eq. (\ref{wFinalExact1}) and formula (\ref{dm1n}) we get the exact expression for the normalized trapping rate 
\begin{align}
	k^{\ast} \left( \epsilon\right) = 2\sum_{m=1}^{\infty}
	\frac{1}
	{\epsilon\lambda_m }\coth\left(\lambda_m/\epsilon \right) \,.	\label{dm1n0}
\end{align}

%%%%%%%%%%%%%%%%%%%%%%%%%%%%%%%%%%%%%%%%%%%%%%%%%%%%%%%%%%%%%%%%%%%%%%%%
\subsection{Small thickness parameter expansions}\label{sec:Asymptotic}
%%%%%%%%%%%%%%%%%%%%%%%%%%%%%%%%%%%%%%%%%%%%%%%%%%%%%%%%%%%%%%%%%%%%%%%%

Let us find here the corresponding  expansion of the local concentration for small thickness parameter $\epsilon$ directly from the exact analytical solution $v \left( \rho,\eta \right)$ (\ref{wFinalExact1}).

First note that the evident relation holds true
\begin{align}
\frac{\sinh\left( \lambda_m\xi/\epsilon \right)}
{\sinh\left( \lambda_m/\epsilon \right)}
\sim
\exp\left(
-{\lambda_m }\eta
\right) \quad \mbox{as} \quad \epsilon \to 0 \,,
\label{v2Asym1}
\end{align}
where for convenience sake we introduced variable
\begin{align}
	\eta := \frac{1-\xi}{\epsilon} \,.
	\label{EtaExact1}
\end{align}
It will further be shown that auxiliary variable $\eta$ is the so-called {\it stretched variable} inside the right diffusion boundary layer (see Subsection~\ref{subsec:RightLayer}).

Using relation (\ref{v2Asym1}) in the exact solution (\ref{wFinalExact1}), one directly obtains the desired asymptotic expansion  
\begin{align}
v \left( \rho,\eta \right)
\sim
2\sum_{m=1}^{\infty}
\frac{J_0 \left( \lambda_m\rho \right)}
{\lambda_m J_1 \left( \lambda_m \right)}
e^{-\lambda_m\eta} \quad \mbox{as} \quad \epsilon \to 0 \,.
\label{ExactLeadingEtaW1}
\end{align}

Clearly asymptotic expression for the trapping rate $k^{\ast} \left( \epsilon\right)$ as $\epsilon \to 0$ coincides with that in the general case $\epsilon > 0$ (\ref{dm1n0}).

%%%%%%%%%%%%%%%%%%%%%%%%%%%%%%%%%%%%%%%%%%%%%%%%%%%%%%%%%%%%%%%%%%%%%%%%%%
\section{Formulation as a singular perturbed problem}\label{sec:singular}
%%%%%%%%%%%%%%%%%%%%%%%%%%%%%%%%%%%%%%%%%%%%%%%%%%%%%%%%%%%%%%%%%%%%%%%%%

%%%%%%%%%%%%%%%%%%%%%%%%%%%%%%%%%%%%%%%%%%%%%%%%%%%%%%%%%%%%%%
\subsection{Governing equations}\label{subsec:State}
%%%%%%%%%%%%%%%%%%%%%%%%%%%%%%%%%%%%%%%%%%%%%%%%%%%%%%%%%%%%%%%

Before proceeding further recall the following 
\begin{definition}\label{defin1}
	Domains are said to be thin if at least one its characteristic lengths is much smaller than the others~\cite{Dzhavadov65}.
\end{definition}

In this section we consider the case of {\it thin cylinder} when the radius of the cylinder $R$ is much smaller than its length $L$, i.e. 
\begin{align}
	0< \epsilon = \frac{R}{L} \ll 1 \,.
	\label{epsDef1}
\end{align}
The asymptotic analysis below is based on the small geometric parameter $\epsilon$ (\ref{epsDef1}) associated with the thin cylindrical domain $\Omega_{\epsilon}$. Besides, note that definition of small parameter $\epsilon $ (\ref{epsDef1}) do not coincide with that given in Ref.~\cite{Traytak14}.

Now we proceed to the asymptotic analysis of the above boundary-value problem as $\epsilon \to 0$.
Then the boundary value problem (\ref{BVP1}), (\ref{AxisRegExact1}), (\ref{ComplBCwall1}), (\ref{ComplBCright1}) takes the form
\begin{eqnarray}
	\frac{1}{\rho}
	\frac{\partial}{\partial \rho}
	\left( \rho\frac{\partial v_\epsilon}{\partial \rho}
	\right)	+	\epsilon^2\frac{\partial^2 v_\epsilon}{\partial \xi^2} = 0 \quad \text{in} \quad \Omega_{\epsilon} \,, 
	\label{ABVP1} \\
	\left. v_\epsilon \right|_{\rho=1} = 0 \,,
	\label{ABVP2}\\
	\left. v_\epsilon \right|_{\xi=0} = 0 \,, \quad \left. v_\epsilon \right|_{\xi=1} = 1 \,,
	\label{ABVP3}\\
	\left. v_\epsilon \right|_{\rho=0} < \infty \,,
	\label{ABVP4}\\
	\left. \frac{\partial v_\epsilon}{\partial \rho} \right|_{\rho=0} = 0 \,.
	\label{ABVP5}
\end{eqnarray}
\begin{remark}
Note that hereafter we designate the desired solution to the above problem when $\epsilon \ll 1$ by $v_\epsilon$ as it is common in the perturbation theory~\cite{Panasenko24}.
\end{remark}

Let us recast Eq.~(\ref{ABVP1}) in the operator form
\begin{align}
	\mathcal{L}_\epsilon v_\epsilon = 0 
	\quad	\text{in} \quad \Omega_{\epsilon}	\,.
	\label{OperForm1}
\end{align}
Here we define the elliptic operator $\mathcal{L}_\epsilon = \mathcal{L}_0 + \epsilon^2\mathcal{L}_1$ consisting of the unperturbed and perturbed operators, respectively:
\begin{eqnarray}	
\mathcal{L}_0
&:=&
-\frac{1}{\rho }\frac{\partial }{\partial \rho }\left(
\rho \frac{\partial  }{\partial \rho }\right) \,, \label{LimOper}\\
\mathcal{L}_1
&:=&
-\frac{\partial^2}{\partial \xi^2} \,.	\label{PerOper}
\end{eqnarray}	
 
One can see that the perturbation parameter in $\mathcal{L}_\epsilon$ has the second order and since in Eq.~(\ref{ABVP1}) the small parameter $\epsilon$ stands at the highest derivative with respect to $\xi$ the problem~(\ref{ABVP1})-(\ref{ABVP5}) is singularly perturbed~\cite{Cole96}. 

It is known that the asymptotic study of the singularly perturbed problem~(\ref{OperForm1}), (\ref{ABVP2})-(\ref{ABVP5}) drastically depends on the spectral properties of the unperturbed operator $\mathcal{L}_0$, therefore, we consider this point first.

\subsection{Spectrum of the unperturbed operator}

Let us show that the following result holds true.
\begin{theorem}\label{theorem1}
The unperturbed operator $\mathcal{L}_0$ of the singularly perturbed problem (\ref{ABVP1}), (\ref{ABVP5}) is not on the spectrum.
\end{theorem}
\begin{proof}
Let us suppose that the eigenfunction $v_\lambda$ satisfies the Sturm-Liouville problem
\begin{eqnarray}
		\qquad \LL_0 v_\lambda	=	\lambda\,v_\lambda \left( \rho \right) \,,
		\quad \rho \in \left(0,\, 1 \right) \,, 
		\label{SLP1}\\
		\left. v_\lambda \right|_{\rho=1} = 0 \,,	\quad
		\left. v_\lambda \right|_{\rho=0} < \infty \,, \quad
		\label{SLP2}
\end{eqnarray}
where $\lambda \in {\mathbb R}$ is the spectral parameter. 

Consider problem~(\ref{SLP1}), (\ref{SLP2}) at $\lambda = 0$. Then Eq.~(\ref{SLP1}) takes the simple form 
$$
\LL_0 v_0 = 0
$$
with the general solution
\begin{align}
v_0 \left( \rho \right) = A_1 + A_2 \ln \rho \,,\label{SpecZeroSol1}
\end{align}
where $A_1$ and $A_2$ are some arbitrary real numbers.
The use of the regularity condition (\ref{SLP2}) on the cylinder axis $\rho = 0$ straightforwardly leads to $A_2 = 0$. In turn it follows from the boundary condition at the wall (\ref{SLP2}) that $A_1 = 0$.
	
Thus, eigenvalue problem (\ref{SLP1}), (\ref{SLP2}) has only the trivial solution $v_0 \left( \rho \right) \equiv 0$ associated with the trivial eigenvalue $\lambda = 0$. Therefore, according to the Vishik-Lyusternik main theorem criterion (see p. 21 of Ref.~\cite{Vishik60}) the singularly perturbed problem~(\ref{ABVP1})-(\ref{ABVP5}) is not on the spectrum.
\end{proof}

Theorem \ref{theorem1} gives an important
\begin{corollary}\label{Corollary1}
Since the singularly perturbed problem~(\ref{ABVP1})-(\ref{ABVP5}) appeared to be not on the spectrum one can derive the explicit asymptotic solution $v_a\left(\rho,\xi; \epsilon \right) \sim v_\epsilon \left(\rho,\xi; \epsilon \right)$ as $\epsilon\to 0$.
\end{corollary}

Moreover, Theorem \ref{theorem1} evidently implies also
\begin{corollary}\label{Corollary2}
If the cylinder $\Omega_{\epsilon}$ is a highly elongated, i.e., $R \ll L$ 
local concentration $u_\epsilon$ does not depend on the coordinates:  
\begin{align}
u_\epsilon\left(\rho,\xi\right) \to  u_0\left(\rho,\xi\right) \equiv 1 \quad \mbox{as} \quad \epsilon\to 0 \,.\label{Cor2}
\end{align}
\end{corollary}

%%%%%%%%%%%%%%%%%%%%%%%%%%%%%%%%%%%%%%%%%%%%%%%%%%%%%%%%
\section{Asymptotic solution}\label{sec:Asymptotic}
%%%%%%%%%%%%%%%%%%%%%%%%%%%%%%%%%%%%%%%%%%%%%%%%%%%%%%%

Theorem \ref{theorem1} guarantees us that the asymptotic solution can be found in an explicit form. This asymptotic solution will be performed by the method of boundary functions~\cite{Butuzov87}.

%%%%%%%%%%%%%%%%%%%%%%%%%%%%%%%%%%%%%%%%%%%%%%%
\subsection{Asymptotic domain decomposition}
%%%%%%%%%%%%%%%%%%%%%%%%%%%%%%%%%%%%%%%%%%%%%%%

First let us carry out the appropriate asymptotic decomposition~\cite{Panasenko22} (as $\epsilon \to 0$) of the given space domain $\Omega_\epsilon$ into three parts (see Fig.~\ref{fig:layers}):  
\begin{align}
\Omega_{\epsilon}= \widehat{\Omega}_\epsilon^{(i)}\cup \Omega_{\epsilon}^{\left( o\right)} \cup \widetilde{\Omega}_\epsilon^{(i)} \,,\label{Dec0}
\end{align}
where we introduced {\it outer subdomain} and {\it inner initial boundary layer subdomains}, respectively
\begin{eqnarray}	
	\Omega_{\epsilon}^{\left( o\right)}=\Omega_{\epsilon}\backslash \widehat{\Omega}_\epsilon^{(i)}\cup \widetilde{\Omega}_\epsilon^{(i)}\,,\label{Decomp2}\\
	\widehat{\Omega}_\epsilon^{(i)}:= \left\{0< \xi <\mathcal{O}\left({\epsilon} \right) \right\}\,,\label{Decomp3}\\
	\widetilde{\Omega}_\epsilon^{(i)}:= \left\{1-\xi <\mathcal{O}\left({\epsilon} \right) \right\}\,.\label{Decomp4}
\end{eqnarray}
\begin{figure}[htbp]
\centering
\includegraphics[width=0.45\textwidth]{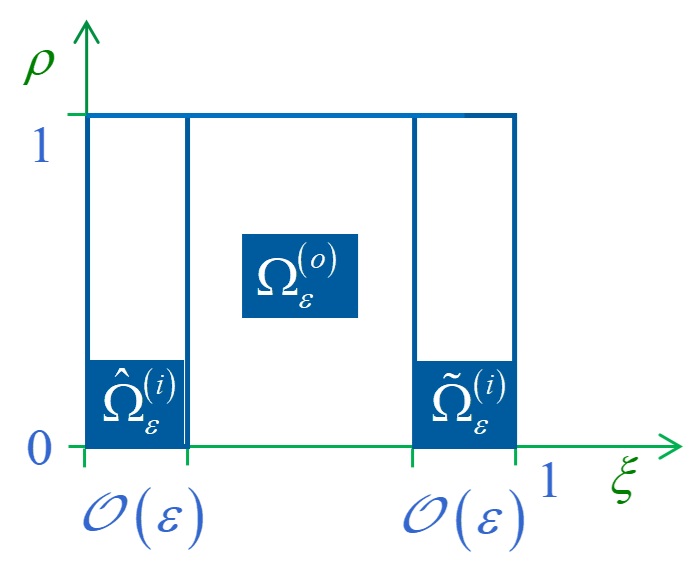}
\caption{Asymptotic decomposition of the domain $\Omega_\epsilon$ into the outer subdomain
		\(\Omega_\epsilon^{(o)}\), the left boundary layer
		\(\widehat{\Omega}_\epsilon^{(i)}\), and the right boundary layer
		\(\widetilde{\Omega}_\epsilon^{(i)}\), both of thickness \(O\left(\epsilon\right)\)
		near \(\xi=0\) and \(\xi=1\), respectively.}
\label{fig:layers}
\end{figure}
Hereafter we use designations: 
\begin{itemize}
\item $\Omega_\epsilon^{(o)}$ is the outer subdomain, away from the ends of the cylinder;
\item $\widehat{\Omega}_\epsilon^{(i)}$ is the spatial boundary layer of thickness
$\mathcal{O}\left(\epsilon\right)$ in the right neighborhood of point $\xi=0$;
\item $\widetilde{\Omega}_\epsilon^{(i)}$ is the spatial boundary layer of thickness
$\mathcal{O}\left(\epsilon\right)$ in the left neighborhood of point $\xi=1 $.
\end{itemize}

\begin{remark}
Note that in applied mathematics the outer subdomain is also known as the regular subdomain.
\end{remark}

It is intuitively clear from Eq. (\ref{ABVP1}) that transverse diffusion is strong ($ \mathcal{O} \left( 1 \right)$), whereas longitudinal diffusion is weak ($\mathcal{O} \left( \epsilon^2 \right)$). The concentration depends only weakly on $\xi$ in the regular subdomain and is mainly determined by the radial part.  In the regular subdomain, as $\epsilon\to 0 $, there exists a limiting solution independent of $\epsilon $. Longitudinal changes are significant only in boundary layers near the ends of the circular cylinder under consideration. Moreover in the boundary layers we should use corresponding  stretched variables, since a regular expansion cannot describe sharp longitudinal changes.

%%%%%%%%%%%%%%%%%%%%%%%%%%%%%%%%%%%%%%%%%%%%%%
\subsection{Outer asymptotic solution}
%%%%%%%%%%%%%%%%%%%%%%%%%%%%%%%%%%%%%%%%%%%%%%

We seek the outer solution in the regular perturbations form
\begin{align}
v_\epsilon^{(o)}\left(\rho,\xi; \epsilon \right)
=
v_0^{(o)}\left(\rho,\xi\right)
+
\epsilon^2 v_1^{(o)}\left(\rho,\xi\right)
+\cdots \,.
\label{WOuterExp}
\end{align}
Substitution of this expansion into problem~(\ref{ABVP1})-(\ref{ABVP5}) gives the boundary value problem with respect to the leading term in the outer expansion (\ref{WOuterExp}): 
\begin{eqnarray}
	\frac{1}{\rho}
	\frac{\partial}{\partial \rho}
	\left(
	\rho
	\frac{\partial v_0^{(o)}}{\partial \rho}
	\right)
	=0 \quad
	\text{in} \quad \Omega_\epsilon^{(o)} \,.
	\label{WOuterEq0} \\
	\left. v_0^{(o)} \right|_{\rho=1} = 0 \,, \quad
	\left. v_0^{(o)} \right|_{\rho=0} < \infty \,,
	\label{WOuterEq1}\\
	\left. \frac{\partial v_0^{(o)}}{\partial \rho} \right|_{\rho=0} = 0 \,,
	\label{WOuterEq2}\\
	\left. v_0^{(o)} \right|_{\xi=0} = 0 \,, \quad \left. v_0^{(o)} \right|_{\xi=1} = 1 \,.
	\label{WOuterEq3}
	\end{eqnarray}
The general solution to the problem (\ref{WOuterEq0})-(\ref{WOuterEq3}) reads
\begin{align}
v_0^{(o)}\left(\rho,\xi\right)
=
A_0\left(\xi\right)+B_0\left(\xi\right)\ln\rho \,,
\label{WOuterGen}
\end{align}
where $A_0\left(\xi\right)$ and $B_0\left(\xi\right)$ are arbitrary functions of $\xi$ to be determined from conditions (\ref{WOuterEq1}).
Using the regularity condition at $\rho=0$ one has $ B_0\left(\xi\right)\equiv0 $, while the wall condition at $\rho=1$  yields $ A_0\left(\xi\right)\equiv 0 $, whereas the axial symmetry condition (\ref{WOuterEq2}) is auto-satisfied. Thus, the leading term in the outer expansion (\ref{WOuterExp}) is identically equal to zero
\begin{align}
v_0^{(o)}\left(\rho,\xi\right) \equiv 0 \,.
\label{WOuterSol}
\end{align}

It is evident that function (\ref{WOuterSol}) obeys the left-end condition, but does not
satisfy the right-end condition in (\ref{WOuterEq3}). 

%%%%%%%%%%%%%%%%%%%%%%%%%%%%%%%%%%%%%%%%%%%%%%%%%%%%%%%%%%%%%%%%%%%
\subsection{Solution in the left boundary layer subdomain}
%%%%%%%%%%%%%%%%%%%%%%%%%%%%%%%%%%%%%%%%%%%%%%%%%%%%%%%%%%%%%%%%%%%

Introducing in the left boundary layer subdomain $\widehat{\Omega}_\epsilon^{(i)}$ the inner stretched variable
\begin{align}
\zeta = \frac{\xi}{\epsilon} 
\label{zetaDefW}
\end{align}
for the left boundary function $\widehat v^{(i)}\left(\rho,\zeta; \epsilon \right)$ we get the inner expansion
\begin{align}
	\widehat v^{(i)}\left(\rho,\zeta; \epsilon \right)
	=
	\widehat v_0^{(i)}\left(\rho,\zeta\right)
	+
	\epsilon^2 \widehat v_1^{(i)}\left(\rho,\zeta\right)
	+\cdots \,.
	\label{Left}
\end{align}
Clearly the leading term in the inner expansion (\ref{Left}) obeys the equation
\begin{align}
\frac{1}{\rho}
\frac{\partial}{\partial \rho}
\left(
\rho
\frac{\partial \widehat v_0^{(i)}}{\partial \rho}
\right)
+
\frac{\partial^2 \widehat v_0^{(i)}}{\partial \zeta^2}
=0 
\quad
\text{in} \quad \widehat{\Omega}_\epsilon^{(i)} \,.
\label{LeftInnerEqW}
\end{align}
Eq. (\ref{LeftInnerEqW}) should be solved under the boundary conditions 
\begin{align}
\left. \widehat v_0^{(i)} \right|_{\zeta=0} = 0 \,, \quad \left. \widehat v_0^{(i)} \right|_{\rho=1} = 0 \,. \label{LeftInnerBCW1}
\end{align}
The matching condition with the outer solution for fixed value $\xi>0$ and $\epsilon \to 0$ from left to right (see Fig.~\ref{fig:layers}) is
\begin{align}
\widehat v_0^{(i)}\left(\rho,\zeta\right)\to v_0^{(o)}\left(\rho,\xi\right) \equiv 0 
\quad \mbox{as} \quad \zeta\to+\infty \,.
\label{LeftMatchW}
\end{align}
Hence, one obtains
\begin{align}
\widehat v_0^{(i)}\left(\rho,\zeta\right) \equiv 0 \,.
\label{LeftInnerTrivialW}
\end{align}
Therefore, the nontrivial boundary function is located only in the left vicinity $\widetilde{\Omega}_\epsilon^{(i)}$ of the right end $\xi=1$.

%%%%%%%%%%%%%%%%%%%%%%%%%%%%%%%%%%%%%%%%%%%%%%%%%%%%%%%%%%%%%%%%%%%%%%%%%%%%%%%%%%%%%
\subsection{Solution in the right boundary layer subdomain}\label{subsec:RightLayer}
%%%%%%%%%%%%%%%%%%%%%%%%%%%%%%%%%%%%%%%%%%%%%%%%%%%%%%%%%%%%%%%%%%%%%%%%%%%%%%%%%%%%%

Completely similar to the previous case introducing  the stretched variable defined by Eq. (\ref{EtaExact1}) in the right boundary layer subdomain $\widetilde{\Omega}_\epsilon^{(i)}$ we can put down the expansion of the right boundary function (\ref{Right}) and problem (\ref{RightInnerEqW}), (\ref{RightInnerBCW2}) with respect to the leading term $\widetilde v_0^{(i)}$ in inner variables $\left(\rho,\eta\right)$:
\begin{eqnarray}
	\widetilde v^{(i)}\left(\rho, \eta; \epsilon \right)
	=
	\widetilde v_0^{(i)}\left(\rho,\eta \right)
	+
	\epsilon^2 \widetilde v_1^{(i)}\left(\rho,\eta\right)
	+\cdots \,.
	\label{Right}\\
	\frac{1}{\rho}
	\frac{\partial}{\partial \rho}
	\left(\rho
	\frac{\partial \widetilde v_0^{(i)}}{\partial \rho}
	\right)
	-
	\frac{\partial^2 \widetilde v_0^{(i)}}{\partial \eta^2}
	=0 
	\quad
	\text{in} \quad \widetilde{\Omega}_\epsilon^{(i)} \,.
	\label{RightInnerEqW}\\
	\left. \widetilde v_0^{(i)} \right|_{\rho=1} = 0 \,, \quad \left. \widetilde v_0^{(i)} \right|_{\eta=0} = 1 \,.
	\label{RightInnerBCW2}	
\end{eqnarray}

The matching condition with the outer solution $v_0^{(o)}$ for fixed value $1-\xi>0$ and $\epsilon \to 0$ from right to left (see Fig.~\ref{fig:layers}) reads
\begin{align}
\widetilde v_0^{(i)}\left(\rho,\eta\right) \to v_0^{(o)}\left(\rho,\xi\right) \equiv 0
\quad \mbox{as} \quad
\eta\to+\infty \,.
\label{RightMatchW}
\end{align}

The required leading term $\widetilde v_0^{(i)}$  may be found as the solution to the boundary value problem (\ref{RightInnerEqW}), (\ref{RightInnerBCW2}) under matching condition (\ref{RightMatchW}) by standard method of separation of variables~\cite{Carslaw59}
\begin{align}
\widetilde v_0^{(i)}\left(\rho,\eta\right)
= 2\sum_{m=1}^{\infty}
\frac{J_0 \left( \lambda_m\rho \right)}
{\lambda_m J_1 \left( \lambda_m \right)}
e^{-\lambda_m\eta} \,.
\label{RightInnerSolW}
\end{align}

\subsection{Composite asymptotic expansion}

According to the general theory of singular perturbations~\cite{Cole96} the leading {\it composite asymptotic solution} for $v_\epsilon$ uniformly valid to order $\mathcal{O}\left(1\right)$ in the whole domain $\Omega_{\epsilon}$ is 
\begin{eqnarray}
v_\epsilon\left(\rho,\xi\right) &\sim& v_0^{(o)}\left(\rho,\xi\right)+ \widehat v_0^{(i)}\left(\rho,\zeta\right) \nonumber \\
 &+&\widetilde v_0^{(i)}\left(\rho,\eta\right) \quad \mbox{as} \quad \epsilon \to 0\,.
\label{CompositeW}
\end{eqnarray}

Simple inspection, taking into account formulas (\ref{WOuterSol}), (\ref{LeftInnerTrivialW}) and (\ref{RightInnerSolW}), shows that this composite asymptotic solution coincides with the leading asymptotic formula obtained from the small thickness parameter expansion of the exact solution (\ref{ExactLeadingEtaW1}). Hence the corresponding asymptotic trapping rate $k^{\ast}_{\epsilon} \left( \epsilon\right)$ also leads to the exact result derived for the general case (\ref{dm1n0}).

%%%%%%%%%%%%%%%%%%%%%%%%%%%%%%%%%%%%%%%%%%%%%%%%%%%%%%%%%%%%%%%%%%%%%
\section{On the 1D reduction approach} \label{sec:1Dreduction}
%%%%%%%%%%%%%%%%%%%%%%%%%%%%%%%%%%%%%%%%%%%%%%%%%%%%%%%%%%%%%%%%%%%%%

Consider now the Fick-Jacobs 1D reduction approach mostly following its presentation given in the monograph~\cite{Dagdug24}. 

Let us begin by introducing 
\begin{definition}\label{defin2}
For the local concentration $v\left(r, z \right)$ within $\Sigma_{z} $ we define its averaging over the cross-section by 
\begin{align}
\langle v\rangle \left(z \right):= \frac{1}{\vert \Sigma_{z} \vert}\int_{0}^{R} 2 \pi r dr v\left(r, z \right) \label{AvEq3D}
\end{align}
\end{definition}
It is important to stress that the average local concentration (\ref{AvEq3D}) commutes with differentiation with respect to $z$.
\begin{remark}
Note that Definition~\ref{defin2} accords with the commonsense view of the term in physics~\cite{Dagdug14} and mathematics~\cite{Sanchez-Palencia80}.
\end{remark} 

Integrating both Eq. (\ref{ABVP1}) and boundary conditions (\ref{ABVP3}) with respect to variable $\rho$ over the interval $\left(0, 1\right)$ with the aid of the boundary condition (\ref{ABVP2}) we can easily arrive at 
\begin{eqnarray}
-\epsilon^2\frac{\partial ^{2}\langle v_{\epsilon}\rangle}{	\partial \xi ^{2}}  = 2 \left. \frac{\partial v_\epsilon}{\partial \rho}\right\vert _{\rho =1} \,,  \label{AvEq1}\\
\left. \langle v_{\epsilon}\rangle \right|_{\xi=0} = 0 \,, \quad \left. \langle v_{\epsilon}\rangle \right|_{\xi=1} = 1\,, \label{AvEq2}
\end{eqnarray}
where the average local concentration is
\begin{align}
\langle v_{\epsilon}\rangle \left(z \right) 
=2\int_{0}^{1}\rho d\rho v_{\epsilon}\left(\rho, \xi \right) \,.\label{AvEq3}
\end{align}

Thus, we obtained the boundary value problem (\ref{AvEq1}), (\ref{AvEq2}) for the ordinary differential equation with respect to $\langle v_{\epsilon}\rangle $.

In turn integration double of Eq. (\ref{AvEq1}) with respect to variable $\xi$ leads to
\begin{eqnarray}
-\frac{1}{2}\epsilon^2 {\langle v_{\epsilon}\rangle}&=&\int_{0}^{\xi} d\zeta_1\int_{0}^{\zeta_1} d\zeta \left. \frac{\partial v_\epsilon}{\partial \rho}\right\vert _{\rho =1} \nonumber \\
 &+& M_1\left(\rho \right)+ M_0\left(\rho \right)\xi\,,
\label{Zwan2}
\end{eqnarray}
where $M_0 \left(\rho \right)$ and $M_1 \left(\rho \right)$ are  arbitrary functions of $\rho $. Hence, using the boundary conditions (\ref{AvEq2}), we derive
$$
-\frac{1}{2}\epsilon^2 {\langle v_{\epsilon}\rangle} =
\int_{0}^{\xi} d\zeta_1\int_{0}^{\zeta_1} d\zeta \left. \frac{\partial v_\epsilon}{\partial \rho}\right\vert _{\rho =1} + M_0\xi\,,
$$
$$
-\frac{1}{2}\epsilon^2 \frac{\partial}{\partial \xi} {\langle v_{\epsilon}\rangle} =
\int_{0}^{\xi} d\zeta \left. \frac{\partial v_\epsilon}{\partial \rho}\right\vert _{\rho =1} + M_0\,,
$$
where constant $M_0$ is
$$
M_0= -\epsilon^2 \frac{1}{2} -
\int_{0}^{1} d\zeta_1\int_{0}^{\zeta_1} d\zeta \left. \frac{\partial v_\epsilon}{\partial \rho}\right\vert _{\rho =1}\,,
$$
Hence it is clear that for the trapping rate the reduction approach yields a compact formula
$$
k^{\ast} \left( \epsilon\right) = - \lim_{\xi \to 1}\frac{\partial}{\partial \xi} {\langle v_{\epsilon}\rangle}
$$
or in an expanded form
$$
k^{\ast} \left( \epsilon\right) = \frac{2}{\epsilon^2}\left[M_0+ \int_{0}^{1}  G_1 \left(\zeta \right)d\zeta\right]\,,
$$
where we used the notation
\begin{align}
\qquad \qquad \qquad G_1 \left(\xi \right):=\left. \frac{\partial v_\epsilon}{\partial \rho}\right\vert _{\rho =1}\,.
\label{FunG1}
\end{align}
This resembles the situation with so-called {\it boundary homogenization approach}, which we discussed in our recent paper~\cite{Traytak26}.

In Ref.~\cite{Dagdug24} it was noted that under the local equilibrium approximation, one can assume proportionality between the original concentration and the reduced density $\langle v_{\epsilon}\rangle \left(\xi \right)$, that is,  
\begin{align}
v_{\epsilon}\left(\rho, \xi \right) \approx \langle v_{\epsilon}\rangle \left(\xi \right) \upsilon \left(\rho \vert \xi \right)\,, \label{Zwan1}
\end{align}
where proportionality is provided by the conditional probability distribution $\upsilon \left(\rho \vert \xi \right)$.

From mathematical point of view (\ref{AvEq3}) is nothing but the averaging procedure, performing projection 2D $\mapsto $ 1D  
$$
v_{\epsilon}\left(\rho, \xi \right) \mapsto \langle v_{\epsilon}\rangle\left(\xi \right)\,,
$$
where the explicit form of the unknown function $\langle v_{\epsilon}\rangle \left(\xi \right)$ is to be determined. 

Thus, in the right-hand side of Eq. (\ref{AvEq1}) one has to attract a supplementary hypothesis to close the reduction with respect to the average function $v_{\epsilon}\left(\rho, \xi \right)$ (\ref{Zwan1}), which is equivalent to setting the function of one variable $G_1 \left(\xi \right)$  (\ref{FunG1}).

%%%%%%%%%%%%%%%%%%%%%%%%%%%%%%%%%%%%%%%%%%%%%%%%%%%%%%%%%
\section{Concluding remarks} \label{sec:Conclusion}
%%%%%%%%%%%%%%%%%%%%%%%%%%%%%%%%%%%%%%%%%%%%%%%%%%%%%%%%%

In this paper we considered the boundary value problem (\ref{BVP1})-(\ref{AxisRegExact1}), which describes coupling of diffusion and reaction inside a circular cylindrical tube. For the thin tube case we derived asymptotic solution of the posed reaction-diffusion problem with respect to the small thickness parameter (\ref{RL1}). This solution was found by means of the boundary functions method, which implies from the general singular perturbation theory. It is important to underline that our derivation of the asymptotic solution was implemented straightforwardly within the scope of the boundary functions method without any additional assumptions. Moreover, used here approach enables us to clarify the physical and mathematical meaning of the obtained results.

The choice of a specific reaction-diffusion problem was caused by its simplicity for asymptotic analysis and also existing of the exact analytical solution. It is significant that obtained asymptotic solution coincides with the exact expansion in terms of the small thickness parameter.

On the other hand the comparison of the boundary function method against know the Fick-Jacobs 1D reduction revealed serious methodological drawbacks of the latter approach. Really, at first glance it seems that the Fick-Jacobs 1D reduction greatly simplifies the task, but immediately a natural and nontrivial question arises concerning a concrete mathematical expression for the average local concentration (\ref{AvEq3}). Obviously it is impossible to resolve this problem without an extra hypothesis about the average local concentration or, equivalently, conditional probability distribution defined by Eq. (\ref{Zwan1}).

Future extension of the present work may include, e.g., generalizations which can take into account the partially absorbing boundary condition on the wall $\partial\Omega_w$ (see Fig.~\ref{Fig1}).

\section*{Acknowledgements}
We sincerely thank Professor L. Dagdug for drawing our attention to Ref.~\cite{Dagdug24}.

%%%%%%%%%%%%%%%%%%%%%%%%%%%%%%%%%%%%%%%%%%%%%%%%%%%%%%%%%
\section{Appendix}
%%%%%%%%%%%%%%%%%%%%%%%%%%%%%%%%%%%%%%%%%%%%%%%%%%%%%%%%%

For the sake of completeness and the readers’ convenience, we recall here a few useful classical mathematical definitions and facts on the Bessel functions.

It is well-known that in a neighborhood of point $\zeta=0$ Bessel's function of the first kind
of order $\nu \ge 0 $ may be defined by its Maclaurin series~\cite{Arfken01}
\begin{align}
	J_{\nu }\left( \zeta \right) =\sum\limits_{m=0}^{\infty }\frac{\left(
		-1\right) ^{m}}{\Gamma \left( m+\nu +1\right) m!}\left( \frac{\zeta }{2}%
	\right) ^{2m+\nu },  \label{Hs3b}
\end{align}
where $\Gamma \left( \beta \right) $ is the gamma function.

The following formulas hold true
\begin{eqnarray}
	\int_{0}^{1} J_0 \left( \lambda\zeta \right) \zeta d \zeta = 
	\frac{J_1\left(\lambda\right)}{\lambda}\,,	\label{IntegralFinal}\\
		\frac{d}{d\zeta}J_0(\zeta)=-J_1(\zeta)\,.
	\label{WatsonDerivativeIdentity}
\end{eqnarray}

Consider now the Bessel functions $J_0\left(\lambda_l\rho\right)$, which play an important role in applications. Hereafter the numbers $\lambda_l$ are the zeros of the transcendental equation
\begin{align}
	J_{0}\left( \lambda _{l}\right) =0\,.  \label{Hs4}
\end{align}
Hence $\left\{ \lambda _{l}\right\}_{l=1}^{\infty }$ is a monotone increasing sequence of positive reals such that, e.g. \cite{Abramowitz} 
$$
\lambda _{1}\approx 2.4048,\, \lambda _{2}\approx 5.5201,\,%
\lambda _{3}\approx 8.6537,...\, 
$$

Finally, let us recall the orthogonality condition
\begin{eqnarray}
	\int_{0}^{1}J_{0}\left( \lambda _{l}\rho\right)J_{0}\left( \lambda _{m}\rho\right)\rho d\rho = \Vert J_{0}\left(\lambda _{l}\rho\right)\Vert_{L_2}^2\delta_{lm}\,,\; \label{PolL3}	\\	
	\Vert J_{0}\left(\lambda_{l}\rho\right)\Vert_{L_2}=\frac{1}{\sqrt2}\frac{d}{d\rho}J_{0}\left( \lambda _{l}\rho\right)=-\frac{\lambda_{l}}{\sqrt2}J_{1}\left( \lambda _{l}\rho\right)\,,  \nonumber
\end{eqnarray}
where $\delta_{lm}$ is the Kronecker delta which is given by $\delta_{lm} =0$ if $l \ne m$ and  $\delta_{lm} =1$ if $l = m$.

\bibliography{LiterF}

\end{document}